\documentclass[12pt]{article}
\usepackage{verbatim,amsmath, amsfonts, amssymb, amsthm, natbib, float, dsfont, bm}
\usepackage{graphicx, algorithmic, algorithm, enumerate,color,hyperref, multirow}
\usepackage{pifont}
\usepackage{xr}
\newtheorem{theorem}{Theorem}
\newtheorem{lemma}{Lemma}

\newtheorem{prop}{Proposition}
\theoremstyle{definition}

\newtheorem{example}{Example}
\newtheorem*{example1}{Example 1 (continued)}

\begin{document}

\title{Necessary and sufficient conditions for  multiple objective 
optimal regression designs}

\author{Lucy L. Gao$^{\circ}$\footnote{Author ordering is alphabetical. Corresponding author: lucy.gao@stat.ubc.ca  }~, Jane J. Ye$^{\ddagger}$, Shangzhi Zeng$^{\ddagger}$, Julie Zhou$^{\ddagger}$ \\~\\
{\small $\circ$ Department of Statistics, University of British Columbia} \\ 
{\small $\ddagger$ Department of Mathematics and Statistics, University of Victoria}  \\
}

\maketitle
\begin{abstract}
We typically construct optimal designs based on a single objective function. To better capture the breadth of an experiment's goals, we could  instead construct a multiple objective optimal design based on multiple objective functions. While algorithms have been developed to find multi-objective optimal designs (e.g. efficiency-constrained and maximin optimal designs), it is far less clear how to verify the optimality of a solution obtained from an algorithm. 
In this paper, we provide theoretical results characterizing optimality for efficiency-constrained and maximin optimal designs on a discrete design space. We demonstrate how to use our results in conjunction with linear programming algorithms to verify optimality. 
\end{abstract}

\vfill

\newpage

\section{Introduction}
Consider modelling the output of a designed experiment as:
\begin{equation}
y_i=f({\bf x}_i, \bm{\theta})+\epsilon_i, \quad i=1, \ldots, n,
\label{Model1}
\end{equation}
where $y_i$ is the response variable observed at design point ${\bf x}_i \in S$ for $S \subseteq \mathbb{R}^p$, $\bm{\theta} \in \mathbb{R}^q$ is a vector of unknown regression parameters, and the $\epsilon_i$'s are independent random errors with $\mathbb{E}[\epsilon_i] = 0$ and $\text{Var}(\epsilon_i) = \sigma^2$. An optimal design chooses the values of ${\bf x}_i$ to answer the experimental questions of interest as precisely as possible. This problem is often formulated in terms of a \emph{single-objective} optimal design problem, where optimality is defined with respect to a single summary measure of the information obtained by fitting a single model to the experimental data. For example, for a particular choice of regression function $f(\cdot, \cdot)$ in \eqref{Model1} and estimator $\hat{\bm \theta}$, an A-optimal design minimizes the average variance of $\hat \theta_1, \ldots, \hat \theta_q$. 

However, experimenters have complex goals that cannot be fully captured by a single-objective optimal design criterion. For example, an experimenter may fit a single model to infer and to predict, and there is little overlap between single-objective optimality criteria that measure inferential versus predictive power. Furthermore, inference on different parameters answers different research questions with varying importance to the experimenter (e.g. main effects are more important than interaction terms or vice versa). Single-objective optimal design criteria do not reflect this variation.  Another consideration is that experimenters may be uncertain about the functional form of the relationship between $y_i$ and ${\bf x}_i$. Thus, they may want a design with good inferential or predictive power for multiple models of the form \eqref{Model1}, rather than a single model of the form \eqref{Model1}.

\emph{Multi-objective} optimal designs combine several single-objective optimal design criteria. Common formulations include the \emph{compound} formulation, which optimizes the weighted sum of the criteria for a set of user-specified weights, the \emph{efficiency-constrained} formulation, which optimizes one criterion while requiring the design efficiency with respect to the other criteria to be higher than user-specified values, and the \emph{maximin} formulation, which maximizes  the minimum efficiency across the set of optimality criteria \citep{kee1999recent, wong2022cvx}. We focus on the efficiency-constrained and maximin formulations, as it is difficult to interpret the practical significance of the weights in the compound formulation. 

Many papers have studied algorithms for finding efficiency-constrained and maximin optimal designs \citep{huang1998sequential, imhof2000graphical, cheng2019multiple}. \citet{wong2022cvx} provide a particularly flexible algorithm: they formulate many efficiency-constrained and maximin optimal design problems as convex optimization problems, then apply an off-the-shelf convex optimization solver (\verb+CVX+, \citealt{cvx}). 

In this paper, we formulate efficiency-constrained and maximin problems as convex optimization problems, along the lines of \citet{wong2022cvx}. We then consider how to verify the optimality of an efficiency-constrained or a maximin optimal design obtained from \verb+CVX+. We provide a complete characterization of optimality for efficiency-constrained and maximin efficiency designs on a discrete design space. Related results appear in the literature for efficiency-constrained optimal designs (see e.g. \citealt{cook1994equivalence} and \citealt{clyde1996equivalence}). Our characterization of optimality for minimax efficiency designs seem to be new, though there is related work on minimax and maximin single-objective optimization problems (see e.g. \citealt{muller1998applications} and \citealt{dette2007maximin}). 

Characterizations of optimality for many commonly-used single-objective optimal design criteria (e.g. $D$- and $A$-) require the optimal design to satisfy a set of easily computable inequalities. By contrast, our characterizations of optimality for efficiency-constrained and maximin designs posit the \emph{existence} of a set of quantities that satisfy a set of equalities and inequalities involving the optimal design. These types of results are thought to be impractical for optimality verification, as it is unclear how to efficiently find a suitable set of quantities. Previous work on efficiency-constrained optimal design problems searches for a suitable set of quantities via grid search and bisection search, but then the computational complexity grows exponentially in the number of objective functions \citep{cheng2019multiple}. Remarkably, we are able to overcome this challenge: it turns out that we can find a suitable set of quantities by solving linear programming problems \citep{luenberger1984linear}. Linear programming problems are a cornerstone of mathematical optimization, and off-the-shelf software are available to solve them accurately and efficiently. 

The rest of the paper is organized as follows. In Section 2, we review concepts related to single-objective optimality criteria, including the necessary and sufficient conditions for optimality. In Sections 3 and 4, we describe how to solve efficiency-constrained and maximin optimal designs and how to verify the optimality of the obtained solutions. Several theoretical results are derived. We apply our approach to several examples in Section 5. The conclusion is in Section 6. Proofs are in the Appendix.

\section{Single-objective optimal designs}

We consider a discrete design space $S_N=\{{\bf u}_1, \ldots, {\bf u}_N \} \subseteq S$ with $N$ points, where ${\bf u}_1, \ldots, {\bf u}_N$ and $S$ are user-specified. If $S$ is a continuous design space, then $S_N$ approximates $S$. We denote a design $\xi$ on $S_N$ by
$\xi({\bf w}) = \left(
\begin{array}{cccc} 
{\bf u}_1 & {\bf u}_2 & \cdots & {\bf u}_N \\
w_1 & w_2 & \cdots & w_N 
\end{array} \right),$
where ${\bf w}$ is an $N$-vector with $i$th entry $w_i$ representing the proportion of design points with value ${\bf u}_i$ for $i = 1, 2, \ldots, N$. Let
$ \Omega \equiv \left \{ {\bf w} \in \mathbb{R}^N: \sum \limits_{i=1}^N w_i = 1, w_i \geq 0 \right \}$.

\subsection{Optimality criteria}
Let ${\bf z}_f({\bf x})$ be the $q$-vector with $j$th entry $\frac{\partial f({\bf x}, \bm \theta)}{\partial \theta_j} \Big |_{\bm \theta = \bm \theta^*}$, where $\bm{\theta}^*$ is the true value of $\bm{\theta}$. The asymptotic covariance matrix of the ordinary least squares estimator of $\bm \theta$ in model \eqref{Model1} with regression function $f(\cdot, \cdot)$ at design $\xi({\bf w})$ is proportional to $\mathcal{I}_f^{-1}({\bf w})$, where 
\begin{eqnarray}
\mathcal{I}_f({\bf w}) = \sum \limits_{i=1}^N ~w_i {\bf z}_f({\bf u}_i) {\bf z}_f^T({\bf u}_i)
\label{Info1}
\end{eqnarray}
is the expected information matrix for model \eqref{Model1} with regression function $f(\cdot, \cdot)$ under the assumption of normally distributed errors. If $f({\bf x}, \bm \theta)$ is non-linear in $\bm \theta$, then $\mathcal{I}_f({\bf w})$ may depend on $\bm \theta^*$. If $\mathcal{I}_f({\bf w})$ depends on $\bm \theta^*$, then optimizing design criteria involving $\mathcal{I}_f({\bf w})$ yields locally optimal designs. In practice, $\bm \theta^*$ is typically unknown, so we must replace it with a ``guess" about its value, e.g. an estimate of $\bm \theta$ from a small pilot study. 

Many single-objective optimal design criteria on $S_N$ can be transformed into convex optimization problems of the form $\underset{{\bf w} \in \Omega}{\min} ~ \Phi({\bf w})$, where $\Phi({\bf w}) = \phi(\mathcal{I}_f({\bf w}))$ for a convex function $\phi$ defined on the set of all $q \times q$ positive definite matrices; see e.g. Table \ref{tab:eff}. Note that the function $\Phi({\bf w})$ is convex as a composition of a convex function and a linear function. We measure the quality of a design ${\bf w}$ using its \emph{efficiency} relative to the optimal design, denoted as $\text{Eff}({\bf w})$. 

\begin{table}[H]
\centering
\caption{\label{tab:eff} Single-objective optimality criteria that solve $\underset{{\bf w} \in \Omega}{\min} ~ \Phi({\bf w})$, where $\Phi({\bf w}) = \phi(\mathcal{I}_f({\bf w}))$ for a  convex function $\phi$ defined on the set of all positive definite matrices. We use $\lambda_{min}({\bf M})$ to denote the smallest eigenvalue of ${\bf M}$. } 
\begin{tabular}{l|lllll}
Criteria & $D$- & $A$- & $c$-, for ${\bf c} \in \mathbb{R}^q$ & $L$-, for ${\bf L} \in  \mathbb{R}^{q \times q'}$ & $E$- \\ 
\hline
$\phi({\bf M})$ & $-\log \det({\bf M})$ & $\text{trace}({\bf M}^{-1})$ & ${\bf c}^T {\bf M}^{-1} {\bf c}$ & $\text{trace}({\bf L}^T {\bf M}^{-1} {\bf L})$ & $-\lambda_{min}({\bf M})$ \\ 
$\text{Eff}({\bf w})$ & $ \left ( \frac{\exp \left \{ \underset{{\bf w}' \in \Omega}{\min}  \Phi({\bf w}') \right \} }{\exp \left \{\Phi({\bf w}) \right \}} \right )^{1/q}$ & $\frac{\underset{ {\bf w}' \in \Omega}{\min}  \Phi({\bf w}')}{\Phi({\bf w})}$ & $\frac{\underset{{\bf w}' \in \Omega}{\min}  \Phi({\bf w}')}{\Phi({\bf w})}$ & $\frac{\underset{{\bf w}' \in \Omega}{\min}  \Phi({\bf w}')}{\Phi({\bf w})}$ & $\frac{\Phi({\bf w})}{\underset{{\bf w}'\in \Omega}{\min}  \Phi({\bf w}')}$
\end{tabular}
\end{table}

The MATLAB-based package \verb+CVX+ \citep{cvx} is a user-friendly option for solving a special subclass of convex optimization problems that includes the convex optimization problems described in Table 1; more details on \verb+CVX+ are provided in Section 3.1. The \verb+CVX+ package has previously been applied to solve many single-objective optimal design problems; see e.g. \citet{gao2017d} and \citet{wong2019cvx}.

\subsection{Necessary and sufficient conditions for optimality}
\label{sec:single-cond}
All of the criteria in Table \ref{tab:eff} lead to convex objective functions, but these objective functions are not all differentiable everywhere. For example, the $E$-optimal design criterion leads to a convex objective function that is non-differentiable at designs ${\bf w}$ such that the smallest eigenvalue of $\mathcal{I}_f({\bf w})$ has geometric multiplicity greater than one. Thus, optimality conditions in this setting rely on \emph{subdifferentials}, 
which generalize derivatives to the class of convex functions. 
We denote the subdifferential of a convex function $\Phi:\mathbb{R}^N\rightarrow \mathbb{R}$  at a point ${\bf w} $  as  $\partial \Phi({\bf w})$. 
The following result describes basic properties of subdifferentials (Chapter 2, \citealt{mordukhovich2013easy}).
\begin{lemma} 
\label{lem:sub} 
Suppose that $\Phi$ and $\Phi'$ are two finite-valued convex functions defined on {$\mathbb{R}^N$}. Then, for any ${\bf w} \in \mathbb{R}^N$:
\begin{enumerate} 
\item If $\Phi$ is differentiable at ${\bf w}$, then $\partial \Phi({\bf w}) = \{\nabla \Phi({\bf w})\}$, where $\nabla \Phi({\bf v})$ is the $N$-vector with $i$th entry  $\frac{\partial \Phi}{\partial w_i} \big |_{{\bf w} = {\bf v}}$. 
\item If $a \geq 0$, then $\partial (a \Phi)({\bf w}) = a \partial \Phi({\bf w}) \equiv \{ a{\bf g}: {\bf g} \in \partial \Phi({\bf w})\}$. 
\item $\partial (\Phi + \Phi')({\bf w}) =  \partial \Phi({\bf w}) + \partial \Phi'({\bf w}) \equiv \{ {\bf g} + {\bf g}': {\bf g} \in \partial \Phi({\bf w}), {\bf g'} \in \partial \Phi'({\bf w}) \}. $
\end{enumerate} 
\end{lemma} 

Let ${\bf e}_i$ denote the $N$-vector with $i$th entry equal to 1 and all other entries equal to 0. The following result characterizes optimality for convex single-objective optimal design criteria on a discrete design space. 
\begin{theorem} 
\label{thm:single-extend}
Suppose that $\Phi:\mathbb{R}^N\rightarrow \mathbb{R}$ is a convex function. Let ${\bf w}^* \in \Omega$. Then, ${\bf w}^* \in \underset{{\bf w} \in \Omega}{\arg \min} ~ \Phi({\bf w})$ if and only if
\begin{equation} 
\exists ~ {\bf g} \in \partial \Phi({\bf w}^*) \text{ such that } {\bf g}^T ({\bf w}^* - {\bf e}_i) \leq 0 ~ \text{for all } i = 1, 2, \ldots, N. \label{eq:ineq-extend} 
\end{equation}
\end{theorem} 

The rest of this subsection is devoted to results that help us evaluate condition \eqref{eq:ineq-extend} in the special case where $\Phi({\bf w}) = \phi(\mathcal{I}_f({\bf w}))$ for a convex function $\phi$. First, if $\phi$ is convex and differentiable at $\mathcal{I}_f({\bf w}^*)$, then Lemma \ref{lem:sub} says that $\partial \Phi({\bf w}^*) = \{\nabla \Phi({\bf w}^*)\}$, and condition \eqref{eq:ineq-extend} simplifies to: 
\begin{equation} 
\left [ \nabla \Phi({\bf w}^*) \right ]^T ({\bf w}^* - {\bf e}_i) \leq 0 ~ \text{for all } i = 1, 2, \ldots, N.  \label{eq:ineq-diff}
\end{equation} 
The following result follows from the matrix chain rule (Section 2.8.1 of \citealt{petersen2012matrix}), and characterizes the left-hand-side of \eqref{eq:ineq-diff}.  
\begin{lemma} 
\label{lem:other}
Let ${\bf w}^* \in \Omega$.  If $\Phi({\bf w}) = \phi(\mathcal{I}_f({\bf w}))$ for a convex function $\phi$ and a regression function $f(\cdot, \cdot)$, and $\phi$ is differentiable at $\mathcal{I}_f({\bf w}^*)$, then:
\begin{equation*} 
\left [ \nabla \Phi({\bf w}^*) \right ]^T ({\bf w}^* - {\bf e}_i)  =  d_{\phi, f}({\bf u}_i, {\bf w}^*), \quad \text{for all } i = 1, 2, \ldots, N, 
\end{equation*} 
where for all $i = 1, 2, \ldots, N$, we define 
\begin{equation} 
d_{\phi, f}({\bf u}_i, {\bf w}^*) \equiv \text{trace} \left ( [ \nabla \phi (\mathcal{I}_f({\bf w}^*) )]^T \left [\mathcal{I}_f({\bf w}^*) - {\bf z}_f({\bf u}_i) {\bf z}_f^T({\bf u}_i)  \right ]  \right ), \label{eq:defd}
\end{equation} 
and where $\nabla \phi ({\bf M}^*)$ is the $q \times q$ matrix with $(j, j')th$ entry $\frac{\partial \phi}{\partial M_{jj'}} \big |_{{\bf M} = {\bf M}^*}$. 
\end{lemma} 
It follows from Theorem \ref{thm:single-extend} and Lemma \ref{lem:other} that characterizing optimality for single-objective optimality criteria with $\Phi({\bf w}) = \phi(\mathcal{I}_f({\bf w}))$ for a differentiable convex function $\phi$ amounts to checking if $d_{\phi, f}({\bf u}_i, {\bf w}^*) \leq 0$ for all $i = 1, 2, \ldots, N$. Furthermore, $d_{\phi, f}({\bf u}_i, {\bf w}^*)$ is straightforward to compute given the formula for $\nabla \phi ({\bf M}^*)$; see Table \ref{tab:phi}. 
\begin{table}[H]
\centering
\caption{\label{tab:phi} Optimality criteria that solve $\underset{{\bf w} \in \Omega}{\min} ~ \phi({\bf A}_f({\bf w}))$ for all of the differentiable convex functions $\phi$ in Table \ref{tab:eff}. Formulas for $\nabla \phi ({\bf M})$ are from \citet{petersen2012matrix}. } 

\begin{tabular}{l|llll}
Criteria & $D$- & $A$- & $c$-, for ${\bf c} \in \mathbb{R}^q$ & $L$-, for ${\bf L} \in  \mathbb{R}^{q \times q'}$ \\ 
\hline
$\phi({\bf M})$ & $-\log \det({\bf M})$ & $\text{trace}({\bf M}^{-1})$ & ${\bf c}^T {\bf M}^{-1} {\bf c}$ & $\text{trace}({\bf L}^T {\bf M}^{-1} {\bf L})$ \\
$\nabla \phi ({\bf M})$ & $-{\bf M}^{-1}$ & $-{\bf M}^{-2}$ & $-{\bf M}^{-1} {\bf c} {\bf c}^T {\bf M}^{-1}$ & $-{\bf M}^{-1} {\bf L} {\bf L}^T {\bf M}^{-1}$ 
\end{tabular}
\end{table}
Combining Theorem \ref{thm:single-extend}, Lemma \ref{lem:other}, and Table \ref{tab:eff} yields various classical equivalence theorems on a discrete design space; see e.g. \citet{kiefer1974general}. 

In the case of $E$-optimality, we have $\Phi({\bf w}) = \phi(\mathcal{I}_f({\bf w}))$ for a \emph{non-differentiable} convex function $\phi$ (Table \ref{tab:eff}), and so Lemma \ref{lem:other} does not always apply. We address this issue with the following result. 
\begin{lemma} 
\label{prop:eopt}
Suppose that $\Phi({\bf w}) = -\lambda_{min}(\mathcal{I}_f({\bf w}))$ for a regression function $f(\cdot, \cdot)$. Let ${\bf w}^* \in \Omega$ and let $r^*$ be the geometric multiplicity of $\lambda_{min}(\mathcal{I}_f({\bf w}^*))$. 
\begin{enumerate} 
\item If $r^* = 1$, then $\partial \Phi({\bf w}^*) = \{ \nabla \Phi({\bf w}^*) \}$, and 
$$ [\nabla \Phi({\bf w}^*)]^T ( {\bf w}^* - {\bf e}_i) =  [({\bf v}^*)^T  {\bf z}_f({\bf u}_i) ]^2 -  \lambda_{min}(\mathcal{I}_f({\bf w}^*)), \text{ for all } i = 1, 2, \ldots, N, $$
where ${\bf v}^*$ denotes an arbitrary unit eigenvector associated with $\lambda_{min}(\mathcal{I}_f({\bf w}^*))$.
\item If $r^* > 1$, then for any ${\bf g} \in \partial \Phi({\bf w}^*)$, there exist $a_1, \ldots, a_{r^*} \geq 0$ such that $\sum_{j=1}^{r^*} a_j = 1$ and 
$${\bf g}^T ( {\bf w}^* - {\bf e}_i) = d_{-\lambda_{min}, f, {\bf a}}({\bf u}_i, {\bf w}^*), \text{ for all } i = 1, 2, \ldots, N,$$ 
where for all $ i = 1, 2, \ldots, N$ we define
\begin{equation} d_{-\lambda_{min}, f, {\bf a}}({\bf u}_i, {\bf w}^*) \equiv \sum \limits_{j=1}^{r^*} a_j   [({\bf v}_j^*)^T  {\bf z}_f({\bf u}_i) ]^2 -  \lambda_{min}(\mathcal{I}_f({\bf w}^*)),\label{eq:defde}
\end{equation} 
where ${\bf v}_1^*, \ldots, {\bf v}_{r^*}^*$ denotes an arbitrary set of orthonormal eigenvectors associated with $\lambda_{min}(\mathcal{I}_f({\bf w}^*))$. 
\end{enumerate} 
\end{lemma} 
Combining Theorem \ref{thm:single-extend} with Lemma \ref{prop:eopt} yields the equivalence theorem for $E$-optimality on a discrete design space \citep{kiefer1974general}.

\section{Efficiency-constrained optimal designs}

Suppose that an experimenter is primarily interested in optimizing one particular single-objective optimality criterion $\Phi_1$ without losing too much efficiency with respect to other criteria $\Phi_2, \ldots, \Phi_K$, for $K \geq 2$. Let $\text{Eff}_k({\bf w})$ denote the efficiency of the design $\xi({\bf w})$ with respect to criterion $\Phi_k$, for $k = 1, \ldots, K$. Given experimenter-specified constants $m_2, \ldots, m_K \in (0, 1)$, an \emph{efficiency-constrained optimal design} on $S_N$ solves:
\begin{eqnarray}
\left\{ \begin{array}{l}
\underset{\bf w \in \mathbb{R}^N}{\min} ~~~~~~~~~~~~~\Phi_1({\bf w}) \\
\mbox{subject to:}  
~~~ ~ \text{Eff}_k({\bf w})  \ge m_k, ~k=2, \ldots, K, 
\\ ~~~~~~~~~~~~~~
~~~~ \sum \limits_{i=1}^N w_i = 1,  w_i \geq 0, ~ i = 1, 2, \ldots, N. 
\end{array}
\right \}.
\label{EffProb1}
\end{eqnarray}
All designs may fail to satisfy the constraints in (\ref{EffProb1}) when the desired minimum efficiencies $m_2, \ldots, m_K$ are large. 

\subsection{Convex optimization problem} 

Suppose that $\Phi_k({\bf w}) = \phi_k(\mathcal{I}_{f_k}({\bf w}))$ for all $k = 1, \ldots, K,$
where $\phi_1, \ldots, \phi_K$ are continuous convex functions chosen from Table \ref{tab:eff}, and $\mathcal{I}_{f_k}({\bf w})$ in \eqref{Info1} is the expected information matrix for a model of the form \eqref{Model1} with regression function $f_k({\bf x}, {\bm \theta})$ for $\bm \theta \in \mathbb{R}^{q_k}$ at design $\xi({\bf w})$. We use the definitions of $\text{Eff}_k({\bf w})$ in Table \ref{tab:eff} to  rewrite \eqref{EffProb1} as the following problem:
\begin{eqnarray}
\left\{ \begin{array}{l}
\underset{\bf w \in \mathbb{R}^N}{\min} ~~~~~~~~~~~~~\Phi_1({\bf w}) \\
\mbox{subject to:} ~~~~~\Phi_k({\bf w}) \le h_k(m_k), ~k=2, \ldots, K, \\ 
~~~~~~~~~~~~~~~~~~ \sum \limits_{i=1}^N w_i = 1, ~ w_i \geq 0, ~ i = 1, 2, \ldots, N. 
\end{array}
\right \}, 
\label{EffProb2}
\end{eqnarray} 
where we define $h_k(m)$ as
\begin{eqnarray} 
h_k(m) = \begin{cases} \left ( \underset{{\bf w}' \in \Omega}{\min}~ \Phi_k({\bf w}') \right ) - q_k \log(m), &\text{ if } \Phi_k({\bf w}) = -\log\det(\mathcal{I}_{f_k}({\bf w})),\\
m \left ( \underset{{\bf w}'\in \Omega}{\min} ~ \Phi_k({\bf w}') \right ), & \text{ if } \Phi_k({\bf w}) = -\lambda_{min}(\mathcal{I}_{f_k}({\bf w})), \\
\frac{1}{m} \left ( \underset{{\bf w}'\in \Omega}{\min} ~ \Phi_k({\bf w}') \right ), & \text{ otherwise}. 
\end{cases} \label{eq:defh}
\end{eqnarray} 
For all $k = 1, 2, \ldots, K$, $\Phi_k({\bf w}) =  \phi_k(\mathcal{I}_{f_k}({\bf w}))$ is a convex function, because $\phi_k$ is a convex function and $\mathcal{I}_{f_k}$ is a linear function.  Thus, \eqref{EffProb2} is a convex optimization problem. Note that our formulation differs from Wong and Zhou (2022) because we use $\Phi_k({\bf w})= - \log \det (\mathcal{I}_{f_k}({\bf w}))$ for $D$-optimality rather than $\Phi_k({\bf w}) = \left [ \det (\mathcal{I}_{f_k}({\bf w}) \right ]^{-1/q_k}$. 

In fact, \eqref{EffProb2} is a convex optimization problem that can be solved by \verb+CVX+  \citep{cvx}, a MATLAB-based package that works with a special subclass of optimization problems; see \citet{grant2008graph} for details on this subclass. \verb+CVX+ automatically converts \eqref{EffProb2} to a form solvable by a numerical convex optimization solver (e.g. SDPT3 or SeDuMi), then translates the numerical results back to the original form.

\subsection{Necessary and sufficient conditions}
The following result characterizes optimality for (\ref{EffProb2}), under the assumption that the minimum efficiency inequality constraints can be strictly satisfied.
\begin{theorem}
\label{Th1}
Suppose that there exists ${\bf w} \in \Omega$ satisfying $\text{Eff}_k({\bf w}) > m_k$ for all $k=2, \ldots, K$. Let ${\bf w}^*$ be a feasible solution for problem \eqref{EffProb2}. Then, ${\bf w}^*$ solves problem (\ref{EffProb2}) if and only if there exists $\eta_2, \ldots, \eta_K \geq 0$ such that 
\begin{enumerate}
\item  $\eta_k \left ( \Phi_k( {\bf w}^*) - h_k(m_k) \right ) = 0$ for all $k = 2, \ldots, K$, and 
\item $ {\bf w}^* \in \underset{{\bf w} \in \Omega}{\arg \min} ~  \left [ \Phi_1({\bf w})  + \sum \limits_{k=2}^K \eta_k \Phi_k({\bf w}) \right ]$. 
\end{enumerate} 
\end{theorem}
Theorem \ref{Th1} is related to results in \citet{lee1988}, \citet{cook1994equivalence}, and \citet{clyde1996equivalence}. 

We will now discuss how to use the results in Section \ref{sec:single-cond} to rewrite Theorem \ref{Th1}. First, suppose that $\phi_1, \ldots, \phi_K$ all correspond to $D$-, $A$-, $c$-, or $L$-optimality. Then, it follows from Theorem \ref{thm:single-extend} and Lemma \ref{lem:other} that we can replace Condition 2 in Theorem \ref{Th1} with:
\begin{equation} 
d_{\phi_1, f_1}({\bf u}_i, {\bf w}^*) + \sum \limits_{k=2}^K \eta_k d_{\phi_K, f_K}({\bf u}_i, {\bf w}^*) \leq 0, \text{ for all } i = 1, 2, \ldots, N,  \label{eq:effd}
\end{equation}
for $d_{\phi, f}({\bf u}_i, {\bf w}^*)$ defined in \eqref{eq:defd}. Table \ref{tab:phi} provides formulas for $\nabla \phi_k({\bf M})$.

Otherwise, we must apply Theorem \ref{thm:single-extend} with Lemma \ref{lem:sub} and Lemma \ref{prop:eopt} to rewrite Condition 2 in Theorem \ref{Th1}. We will illustrate via example. 

\begin{example} 

Suppose that $\phi_1({\bf M}) = -\lambda_{min}({\bf M})$, and $\phi_2, \ldots, \phi_K$ all correspond to $D$-, $A$-, $c$-, or $L$-optimality. Then, it follows from Theorem \ref{thm:single-extend} and Lemma \ref{lem:sub} that Condition 2 in Theorem \ref{Th1} holds if and only if  there exists ${\bf g} \in \partial \Phi_1({\bf w}^*)$ such that
\begin{equation} {\bf g}^T \left ( {\bf w}^* - {\bf e}_i \right ) + \sum \limits_{k=2}^K \eta_k d_{\phi_k, f_k}({\bf u}_i, {\bf w}^*) \leq 0 \text{ for all } i = 1, 2, \ldots, N, \label{eq:Econd0}
\end{equation} 
for $d_{\phi, f}({\bf u}_i, {\bf w}^*)$ defined in \eqref{eq:defd}.  Let $r^*$ be the geometric multiplicity of $\lambda_{min}(\mathcal{I}_{f_1}({\bf w}^*))$. It further follows from Lemma \ref{prop:eopt} that: 
\begin{itemize} 
\item \underline{Case 1 ($r^* = 1)$}: Condition 2 in Theorem \ref{Th1} holds if and only if 
\begin{equation} 
[({\bf v^*})^T {\bf z}_{f_1}({\bf u}_i, {\bf w}^*) ]^2 - \lambda_{min}(\mathcal{I}_{f_1}({\bf w}^*)) + \sum \limits_{k=2}^K \eta_k d_{\phi_k, f_k}({\bf u}_i, {\bf w}^*) \leq 0, \forall ~ i = 1, \ldots, N,  \label{eq:Econdspecial} 
\end{equation} 
where ${\bf v^*}$ denotes an arbitrary unit eigenvector of $\lambda_{min}(\mathcal{I}_{f_1}({\bf w}^*))$. 
\item \underline{Case 2 ($r^* > 1)$}: Condition 2 in Theorem \ref{Th1} holds if and only if  there exists $a_1, \ldots, a_{r^*} \geq 0$ such that $\sum_{j=1}^{r^*} a_j = 1$ and 
\begin{equation} 
d_{-\lambda_{min}, f_1, {\bf a}}({\bf u}_i, {\bf w}^*) + \sum \limits_{k=2}^K \eta_k d_{\phi_k, f_k}({\bf u}_i, {\bf w}^*) \leq 0, ~ \forall ~ i = 1, \ldots, N, \label{eq:Econd1}
\end{equation} 
where $d_{-\lambda_{min}, f, {\bf a}}({\bf u}_i, {\bf w}^*)$ is defined in \eqref{eq:defde}. 
\end{itemize} 
\end{example}

\subsection{Optimality verification via linear programming}
Suppose that we have obtained ${\bf w}^*$ by solving \eqref{EffProb2} numerically, where $\phi_1, \ldots, \phi_K$ all correspond to $D$-, $A$-, $c$-, or $L$-optimality. We know from Theorem \ref{Th1} that ${\bf w}^*$ is optimal if we can find $\eta_2, \ldots, \eta_K \geq 0$ such that Conditions 1--2 in Theorem \ref{Th1} are satisfied. However, ${\bf w}^*$ is an approximate numerical solution, and is thus unlikely to  satisfy Conditions 1--2 exactly. Instead, we check whether ${\bf w}^*$ is ``close enough" to optimal by searching for $\eta_2, \ldots, \eta_K \geq 0$ such that: 
\begin{gather}  
\eta_k \left ( \Phi_k( {\bf w}^*) - h_k(m_k) \right ) \leq \delta, \quad k = 2, \ldots, K, \label{eq:Cond1a} \\ 
-\eta_k \left ( \Phi_k( {\bf w}^*) - h_k(m_k) \right ) \leq \delta, \quad k = 2, \ldots, K, \label{eq:Cond1b} \\ 
d_{\phi_1, f_1}({\bf u}_i, {\bf w}^*) + \sum \limits_{k=2}^K \eta_k d_{\phi_k, f_k}({\bf u}_i, {\bf w}^*) \leq \delta, \quad i = 1, 2, \ldots, N, \label{eq:Cond2}
\end{gather} 
where $\delta$ is a small positive constant (e.g. $\delta = 10^{-4}$). Here, \eqref{eq:Cond1a}--\eqref{eq:Cond1b} relax Condition 1 in Theorem \ref{Th1}, and \eqref{eq:Cond2} relax Condition 2 in Theorem \ref{Th1}, and $\delta$ controls our definition of ``close enough" to optimal. Similar ideas appear in single-objective optimal designs \citep{wong2019cvx}.

We propose solving the following optimization problem: 
\begin{eqnarray}
\left\{ \begin{array}{l}
\underset{\bm \eta \in \mathbb{R}^{K-1}}{\min} ~~~ ~{\bf 1}_{K-1}^T \bm \eta \\
\mbox{subject to:} ~{\bm \eta} \ge \bm 0_{K-1}, 
~~{\bf B}_1^T  {\bm \eta} \le {\bf b}_1, 
~~{\bf C}_1 {\bm \eta} \le \delta {\bf 1}_{K-1},  -{\bf C}_1 {\bm \eta} \le \delta {\bf 1}_{K-1}.
\end{array}
\right \},
\label{LProg1}
\end{eqnarray} 
where $\le$ and $\ge$ denote component-wise inequality,   ${\bf 1}_{K-1}$ is the $(K-1)$-vector with every entry equal to 1,  ${\bf B}_1$ is the $(K-1) \times N$ matrix with $(k, i)$th entry $d_{\phi_{k+1}, f_{k+1}}({\bf u}_i, {\bf w}^*)$ where $d_{\phi, f}({\bf u}_i, {\bf w}^*)$ is defined in \eqref{eq:defd} for $k = 1, 2, \ldots, K-1$, ${\bf b}_1$ is the $N$-vector with $i$th entry equal to $\delta - d_{\phi_1, f_1}({\bf u}_i, {\bf w}^*)$, and ${\bf C}_1 = \text{diag}(\Phi_2({\bf w}^*) - h_2(m_2), \ldots, \Phi_K({\bf w}^*) - h_K(m_K)).$

If we are able to find a solution $\bm \eta^*$ to \eqref{LProg1}, then we know that ${\bf w}^*$ and $\bm \eta^*$ jointly satisfy \eqref{eq:Cond1a}--\eqref{eq:Cond2}. This would mean that the conditions in Theorem \ref{Th1} (approximately) hold, therefore ${\bf w}^*$ is indeed optimal for \eqref{EffProb2}. Furthermore, \eqref{LProg1} is a \emph{linear programming problem} \citep{luenberger1984linear}: its objective function and constraints are all linear. Thus, we can solve \eqref{LProg1} by simply applying an off-the-shelf linear programming solver like the \verb+linprog+ function in the Optimization Toolbox of \verb+MATLAB+.

If one or more of $\phi_1, \ldots, \phi_K$ correspond to $E$-optimality, then the following example illustrates that we can still verify the conditions in Theorem \ref{Th1} via linear programming.

\begin{example1} Consider Example 1 in Section 3.2, 
 where  $\phi_1({\bf M}) = -\lambda_{min}({\bf M})$ and $\phi_2, \ldots, \phi_K$ all correspond to $D$-, $A$-, $c$-, or $L$-optimality. Recall that we defined $r^*$ to be the geometric multiplicity of $\lambda_{min}(\mathcal{I}_{f_1}({\bf w}^*))$. We previously showed that if $r^* = 1$, then Condition 2 in Theorem \ref{Th2} is equivalent to \eqref{eq:Econdspecial}, which defines a set of $N$ linear equalities in $\eta_2, \ldots, \eta_K$ Thus, we can minimize $\sum \limits_{k=2}^K \eta_k$ subject to the $2(K-1)$ linear inequalities defined in \eqref{eq:Cond1a}--\eqref{eq:Cond1b} and the following relaxed version of \eqref{eq:Econdspecial},
\begin{equation*} 
[({\bf v^*})^T {\bf z}_{f_1}({\bf u}_i, {\bf w}^*) ]^2 - \lambda_{min}(\mathcal{I}_{f_1}({\bf w}^*)) + \sum \limits_{k=2}^K \eta_k d_{\phi_k, f_k}({\bf u}_i, {\bf w}^*) \leq \delta, \forall ~ i = 1, \ldots, N.
\end{equation*}  
If we are able to find a solution to this linear programming problem, then we know that ${\bf w}^*$ is an optimal solution. 
 
We also showed in Section 3.2 that if $r^* > 1$, then Condition 2 holds if and only if there exists $a_1, \ldots, a_{r^*} \geq 0$ and $\eta_2, \ldots, \eta_K \geq 0$ such that $\sum \limits_{j=1}^{r^*} a_j = 1$ and \eqref{eq:Econd1} holds. Thus, Theorem \ref{Th1} says that ${\bf w}^*$ is optimal for \eqref{EffProb2} if and only if there exists $a_1, \ldots, a_{r^*} \geq 0$ and  $\eta_2, \ldots, \eta_K \geq 0$ such that $\eta_k(\Phi_k({\bf w}^*) - h_k(m_k)) = 0$ for all $k = 2, \ldots, K$ and  \eqref{eq:Econd1} holds. Observe that \eqref{eq:Econd1} defines $N$ linear inequalities in $\eta_2, \ldots, \eta_K$ and in $a_1, \ldots, a_{r^*}$. Thus, we can minimize $\sum \limits_{k=2}^K \eta_k + \sum \limits_{j=1}^{r^*} a_j$ subject to the $2(K-1)$ linear inequalities defined in \eqref{eq:Cond1a}--\eqref{eq:Cond1b} and the following relaxation of \eqref{eq:Econd1},
\begin{equation*} 
d_{-\lambda_{min}, f_1, {\bf a}}({\bf u}_i, {\bf w}^*) + \sum \limits_{k=2}^K \eta_k d_{\phi_k, f_k}({\bf u}_i, {\bf w}^*) \leq \delta, ~ \forall ~ i = 1, \ldots, N.
\end{equation*}  
 Once again, if we are able to find a solution to this linear programming problem, then we know that ${\bf w}^*$ is an optimal solution. 

\end{example1} 

\section{Maximin optimal designs}

Suppose that an experimenter requires a design that yields reasonable efficiencies for \emph{all of} $K$ single-objective optimality criteria. 
We formulate this \emph{maximin} design problem as:
\begin{eqnarray}
\left\{ \begin{array}{l}
\underset{{\bf w} \in \mathbb{R}^N}{\max} ~~~~~~~~~~~~\min\{ \text{Eff}_1({\bf w}), \ldots, \text{Eff}_K({\bf w}) \} \\
\mbox{subject to:} ~~~~~ \sum \limits_{i=1}^N w_i = 1, ~ w_i \geq 0, ~ i = 1, 2, \ldots, N. 
\end{array}
\right \}.
\label{MaximinProb3}
\end{eqnarray} 

\subsection{Convex optimization problem}
Problem \eqref{MaximinProb3} is hard to solve directly, since the objective function involves a minimization. However, we can equivalently formulate \eqref{MaximinProb3} as:
\begin{eqnarray}
\left\{ \begin{array}{l}
\underset{{\bf w} \in \mathbb{R}^N,t \in \mathbb{R}}{\max} ~~~~1/t \\
\mbox{subject to:}~\text{Eff}_k({\bf w}) \ge 1/t, ~k=1, \ldots, K, \\
~~~~~~~~~~~~~~~~t \ge 0, ~w_i \ge 0, i=1, \ldots, N, ~\sum_{i=1}^N w_i=1 
\end{array}
\right \}.
\label{MaximinProb4}
\end{eqnarray} 
This formulation eliminates the minimization from the objective function by introducing an additional optimization variable ($t$). Furthermore, when $\Phi_k({\bf w}) = \phi_k(\mathcal{I}_{f_k}({\bf w}))$ with $\phi_1, \ldots, \phi_K$ chosen from the convex functions in Table \ref{tab:eff}, \eqref{MaximinProb4} is equivalent to:
\begin{eqnarray}
\left\{ \begin{array}{l}
\underset{{\bf w} \in \mathbb{R}^N,t \in \mathbb{R}}{\min} ~~~~~t \\
\mbox{subject to:}~ \Phi_k({\bf w})  \le h_k(1/t)  , ~k=1, \ldots, K, \\
~~~~~~~~~~~~~~~~t \ge 0, \sum_{i=1}^N w_i=1, 
~w_i \ge 0, i=1, \ldots, N.
\end{array}
\right \},
\label{MaximinProb5}
\end{eqnarray} 
where $h_k(\cdot)$ is defined in \eqref{eq:defh}. This is a convex optimization problem, because $h_k(1/t)$ is a concave function of $t$ and $\Phi_k({\bf w})$ is a convex function of ${\bf w}$. Furthermore, we can solve \eqref{MaximinProb4} using \verb+CVX+, as it fits into the\verb+CVX+ modelling framework described in \citet{grant2008graph}. Note that our formulation is more general than that of \citet{wong2022cvx}, as we allow the user to ``mix-and-match" any combination of the criteria in Table \ref{tab:eff}.

\subsection{Necessary and sufficient conditions for optimality}

The following result characterizes optimality for (\ref{MaximinProb5}). 
\begin{theorem}
Suppose that $({\bf w}^*, t^*)$ are feasible for problem \eqref{MaximinProb5}. Then, $({\bf w}^*, t^*)$ solves (\ref{MaximinProb5}) if and only if there exists $\eta_1, \ldots, \eta_K \geq 0$ satisfying:  
\begin{enumerate}
\item $\sum \limits_{k=1}^K \eta_k \left ( \frac{d}{d_t} h_k(1/t) \Big |_{t = t^*} \right ) = 1$.  
\item $\eta_k (\Phi_k({\bf w}^*) - h_k(1/t^*)) = 0$ for all $k = 1, 2, \ldots, K$.
\item ${\bf w}^* \in \underset{{\bf w} \in \Omega}{\arg \min} \left \{ \sum \limits_{k=1}^K \eta_k \Phi_k({\bf w}) \right \}.$
\end{enumerate}  
\label{Th2}
\end{theorem}

We now discuss how to use the results in Section \ref{sec:single-cond} to rewrite Theorem \ref{Th2}. First, suppose that $\phi_1, \ldots, \phi_K$ all correspond to $D$-, $A$-, $c$-, or $L$-optimality. Then, it follows from Theorem \ref{thm:single-extend} and Lemma \ref{lem:sub} that we can replace Condition 3 in Theorem \ref{Th2} with:
\begin{equation} \sum \limits_{k=1}^K \eta_k d_{\phi_k, f_k} ({\bf u}_i, {\bf w}^*) \leq 0 \text{ for all } i = 1, 2, \ldots, N,\label{eq:mmCond} 
\end{equation} 
where $d_k({\bf u}_i, {\bf w}^*)$ is given in \eqref{eq:defd} and formulas for $\nabla \phi_k({\bf M})$ are given in Table \ref{tab:phi}. Otherwise, we will need to apply Theorem \ref{thm:single-extend} with Lemma \ref{lem:sub} and Lemma \ref{prop:eopt} to rewrite Condition 3. We will illustrate with an example. 

\begin{example} 
Suppose that $\phi_1({\bf M}) = -\lambda_{min}({\bf M})$, and $\phi_2, \ldots, \phi_K$ all correspond to $D$-, $A$-, $c$-, or $L$-optimality.  Then, it follows from Theorem \ref{thm:single-extend} and Lemma \ref{lem:sub} that Condition 3 in Theorem \ref{Th2} holds if and only if  there exists ${\bf g} \in \partial \Phi_1({\bf w}^*)$ such that
\begin{equation} \eta_1 {\bf g}^T \left ( {\bf w}^* - {\bf e}_i \right ) + \sum \limits_{k=2}^K \eta_k d_{\phi_k, f_k}({\bf u}_i, {\bf w}^*) \leq 0, ~\forall~ i = 1, 2, \ldots, N,. \label{eq:mmcond0}
\end{equation} 
Let $r^*$ be the geometric multiplicity of $\lambda_{min}(\mathcal{I}_{f_1}({\bf w}^*))$. Then, based on Lemma \ref{prop:eopt}, we can consider two cases: 
\begin{itemize} 
\item \underline{Case 1 ($r^*$ = 1)}: Condition 3 in Theorem \ref{Th2} holds if and only if 
\begin{equation} 
\eta_1 [({\bf v^*})^T {\bf z}_{f_1}({\bf u}_i, {\bf w}^*) ]^2 - \eta_1\lambda_{min}(\mathcal{I}_{f_1}({\bf w}^*)) + \sum \limits_{k=2}^K \eta_k d_{\phi_k, f_k}({\bf u}_i, {\bf w}^*) \leq \delta, \forall ~ i = 1, \ldots, N,\label{eq:defde-mm}
\end{equation} 
where ${\bf v}^*$ denotes an arbitrary unit eigenvector corresponding to  $\lambda_{min}(\mathcal{I}_{f_1}({\bf w}^*))$. 
\item \underline{Case 2 ($r^* > 1)$}: Condition 3 in Theorem \ref{Th2} holds if and only if there exists $a_1, \ldots, a_{r^*} \geq 0$ such that $\sum \limits_{j=1}^{r^*} a_j = 1$ and 
$$ \eta_1 d_{-\lambda_{min}, f_1, {\bf a}} ({\bf u}_i, {\bf w}^*) + \sum \limits_{k=2}^K \eta_k d_{\phi_k, f_k} ({\bf u}_i, {\bf w}^*) \leq 0  \text { for all }  i = 1, 2, \ldots, N, $$
where $d_{-\lambda_{min}, f_1, {\bf a}} ({\bf u}_i, {\bf w}^*)$ is defined in \eqref{eq:defde}.
\end{itemize} 
\end{example} 

\subsection{Optimality verification via linear programming} 
Suppose that we have obtained a candidate solution $({\bf w}^*, t^*)$ by solving \eqref{MaximinProb5} numerically (e.g. via \verb+CVX+), where $\phi_1, \ldots, \phi_K$ all correspond to $D$-, $A$-, $c$-, or $L$-optimality. Based on the results in Section 4.2, we would like to find 
$\eta_1, \ldots, \eta_K \geq 0$ such that:
\begin{gather}  
\sum \limits_{k=1}^K \eta_k \left ( \frac{d}{dt} h_k(1/t) \Big |_{t = t^*} \right ) = 1, \label{eq:Cond1m}\\ 
\eta_k (\Phi_k({\bf w}^*) - h_k(1/t^*)) \leq \delta, \quad 1 = 2, \ldots, K, \label{eq:Cond2ma} \\ 
-\eta_k (\Phi_k({\bf w}^*) - h_k(1/t^*)) \leq \delta, \quad 1 = 2, \ldots, K, \label{eq:Cond2mb} \\ 
\sum \limits_{k=1}^K \eta_k d_{\phi_k, f_k}({\bf u}_i, {\bf w}^*) \leq \delta, \quad i = 1, 2, \ldots, N, \label{eq:Cond3m}
\end{gather} 
where $\delta$ is a small positive constant. Here, we have relaxed Conditions 2 and 3 in Theorem \ref{Th2} because ${\bf w}^*$ is an approximate solution, along the lines of Section 3.3. We achieve this goal by solving the following linear programming problem using the \verb+linprog+ function in \verb+MATLAB+:
\begin{eqnarray} 
\label{LProg2}
\left\{ \begin{array}{l}
\underset{\bm \eta \in \mathbb{R}^{K}}{\min} ~ ~{\bf 1}_K^T \bm \eta \\
\mbox{subject to:} ~{\bm \eta} \ge \bm 0_{K}, 
~{\bf b}_2^T {\bm \eta} = 1, ~{\bf B}_2^T {\bm \eta} \le \delta {\bf 1}_N, {\bf C}_2 \bm \eta \leq  \delta {\bf 1}_{K}, -{\bf C}_2 \bm \eta \leq  \delta {\bf 1}_{K}. 
\end{array}
\right \},
\end{eqnarray} 
where ${\bf b}_2$ is the $K$-vector with $k$th entry equal to $\left ( \frac{d}{d_t} h_{k}(1/t) \Big |_{t = t^*} \right )$ for $h_k$ defined in \eqref{eq:defh}, ${\bf B}_2$ is the $K \times N$ matrix with $(k, i)$th entry equal to $d_{\phi_k, f_k}({\bf u}_i, {\bf w}^*)$ for $d_{\phi, f}({\bf u}_i, {\bf w}^*)$ defined in \eqref{eq:defd}, and ${\bf C}_2 = \text{diag}( \Phi_1({\bf w}^*) - h_1(m_1), \ldots, \Phi_K({\bf w}^*) - h_K(m_K))$. 

When one or more of $\phi_1, \ldots, \phi_K$ correspond to $E$-optimality, we can still rewrite Condition 3 in Theorem \ref{Th2} as a set of linear inequalities; see e.g. Example 2. Thus, we can still verify the conditions in Theorem \ref{Th2} using linear programming. We omit the details of how to set up the linear programming problem, as the ideas are similar to Example 1 in Section 3.3.

\section{Applications}

In all three of the following applications, we set $\delta = 10^{-4}$ when verifying optimality via linear programming as described in Sections 3.3 and 4.3. Any choice of $\delta$ larger than $10^{-6}$ yields the same results. All computations were performed on a 2021 M1 Macbook Pro with 10 cores and 16 GB memory.  We provide MATLAB code to reproduce all numerical results at \verb+https://github.com/lucylgao/multi-objective-paper-code-2022+. 

\label{sec:app}
{\bf Application 1.} Consider a four-parameter  compartment model of the form \eqref{Model1} with $p = 1$, $q = 4$, $f(x, \bm \theta) = \theta_1 e^{- \theta_2 x} + \theta_3 e^{- \theta_4 x}$, and $S = [0, 15]$, where the responses $y_i$ represent the concentration level of a drug in compartments and $x$ denotes the sampling time. 
This model has been studied in optimal designs for various optimality
criteria, including multi-objective criteria \citep{huang1998sequential, cheng2019multiple}. 

We seek efficiency-constrained optimal designs that solve \eqref{EffProb1} with $\Phi_k({\bf w}) = \phi_k(\mathcal{I}_{f}({\bf w}))$ for $k = 1, 2, 3$. As in \citet{cheng2019multiple}, we let $\phi_1$ correspond to $L$-optimality with ${\bf L} = \text{diag}\left (\frac{1}{\theta_1}, \frac{1}{\theta_2}, \frac{1}{\theta_3}, \frac{1}{\theta_4} \right)$, $\phi_2$ correspond to $D$-optimality, and $\phi_3$ correspond to $L$-optimality with ${\bf L} = \left [ \int_2^{10} {\bf z}_f(x, \bm \theta^*) {\bf z}_f^T(x, \bm \theta^*) dx \right ]^\frac{1}{2}$, where ${\bf z}_f(x, \bm \theta) = (e^{-\theta_2x}, -\theta_1xe^{-\theta_2 x},
e^{-\theta_4 x}, -\theta_3xe^{-\theta_4 x})^\top$ and $\bm \theta^* = (5.25, 1.34, 1.75, 0.13)^\top$. We discretize the continuous design space $S$ to form  $S_N$ with
$u_i=15(i-1)/(N-1)$ for $i=1, \ldots, N$. 

First, we find the single-objective optimal designs by solving $\underset{{\bf w} \in \Omega}{\min} ~\Phi_k({\bf w})$ for each $k = 1, 2, 3$ using \verb+CVX+. Then, we solve \eqref{EffProb1} with $m_2 = 0.9$, $m_3 = 0.8$, and $N = 501$  using \verb+CVX+ to get a solution denoted as ${\bf w}^{*m}$. We report the single-objective optimal designs and ${\bf w}^{*m}$ in Table \ref{table1}. The 
efficiencies at ${\bf w}^{*m}$ are close to those reported in \citet{cheng2019multiple}.

\begin{table}[h!]
\caption{For Application 1, single-objective optimal designs and the efficiency-constrained optimal design with $m_2 = 0.9, m_3 = 0.8$. }
\begin{center}
\begin{tabular}{rrrrrrrr} \hline
$\phi_1$-optimal & $\phi_2$-optimal & $\phi_3$-optimal & efficiency-constrained\\
points (weights) & points (weights)  & points (weights)  & points (weights) \\ \hline 
0 (0.0591) &  0 (0.2500) & 0 (0.1339) & 0 (0.1339)  \\ 
0.6300 (0.1315) & 0.6600 (0.2500) & 0.9600 (0.0663) & 0.6600 (0.1513)\\ 
2.9400  (0.3126) & 2.8800 (0.2500) &  3.300 (0.4502) & 3.0300 (0.2481) \\ 
13.2900 (0.4968) & 11.0100 (0.2441) & 9.7500 (0.2231) & 3.0600 (0.0942)  \\ 
				&  11.0400 (0.0059) & 9.7800 (0.2502) & 10.8300 (0.0807) \\ 
				&                 &                    & 10.8600 (0.2918) \\
    \hline
\end{tabular}
\end{center} 
\label{table1}
\end{table}

We then verify the conditions for optimality in Theorem \ref{Th1} for ${\bf w}^{*m}$ by using the \verb+linprog+ MATLAB function to solve \eqref{LProg1} with $\delta = 10^{-4}$ as described in Section 3.3. Solving \eqref{LProg1} yielded $\eta_2^* = 36.4870$ and $\eta_3^* = 5.0767$. Since we obtain a solution, we know that ${\bf w}^{*m}$ is indeed the efficiency-constrained optimal design (Theorem \ref{Th1}). Figure \ref{fig1} shows that $d_{\phi_1, f}(u_i, {\bf w}^{*m})) + \sum \limits_{k=2}^3 \eta_k^* d_{\phi_k, f}(u_i, {\bf w}^{*m})) \leq \delta$ for all $i = 1, 2, \ldots, N$; this amounts to visually showing that Condition 2 in Theorem \ref{Th1} is satisfied for ${\bf w}^{*m}$, $\eta_2^*$ and $\eta_3^*$. Figure \ref{fig1} also shows that $d_{\phi_k, f}(u_i, {\bf w}^{*m})$ is not uniformly non-negative for all $k = 1, 2, 3$. Thus, ${\bf w}^{*m}$ is not the single-objective optimal design that minimizes $\Phi_1$, $\Phi_2$ or $\Phi_3$ (Theorem \ref{thm:single-extend} and Lemma \ref{lem:other}).

 \begin{figure}[h!]
\centering
\includegraphics[scale=0.9]{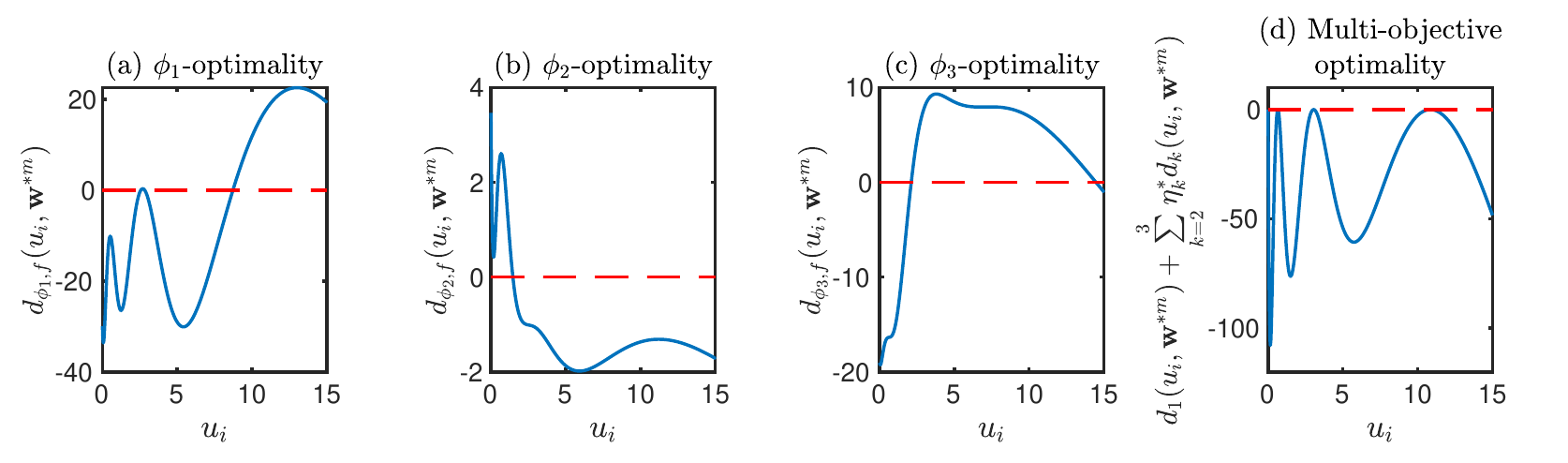}
\caption{For $(m_2, m_3)=(0.9, 0.8)$ in Application 1, panels (a)--(c) display plots of $d_{\phi_k, f}(u_i, {\bf w}^{*m}))$ for $k = 1, 2, 3$, and panel (d) displays $d_{\phi_1, f}(u_i, {\bf w}^{*m})) + \sum \limits_{k=2}^3 \eta_k^* d_{\phi_k, f}(u_i, {\bf w}^{*m}))$. In panel (d), the red dashed line represents the horizontal line $y = \delta$, for $\delta = 10^{-4}.$}
\label{fig1}
\end{figure}

Next, we vary $m_2$ and $m_3$ and look at how the results change; representative results are given in Table \ref{table2}.
For $m_2=0.90$ and $m_3=0.70$, we find that $\eta_3^* = 0$; this is because $\Phi_3({\bf w}^{*m}) < h_3(m_3)$ and $\eta_2^*, \eta_3^*$ satisfy Condition 1 in Theorem \ref{Th1}. Similarly, for $m_2=0.70$ and $m_3=0.70$,  we find that $\eta_2^* = \eta_3^* = 0$ because  $\Phi_k({\bf w}^{*m})< h_k(m_k)$ (i.e. $\text{Eff}_k({\bf w}^{*m})>m_k$) for $k = 2, 3$; 
this implies that the multi-objective optimal design ${\bf w}^{*m}$ is also a single-objective optimal design maximizing $\Phi_1$ (Theorem \ref{thm:single-extend} and Lemma \ref{lem:other}).
For $m_2=0.90$ and $m_3=0.90$, there is no feasible solution.

\begin{table}[h!]
\begin{small}
\caption{For Application 1, efficiencies and $\eta_2^*$ and $\eta_3^*$, for various $(m_2, m_3)$.}
\begin{center}
\begin{tabular}{ccccc} \hline
Case ($m_2,m_3$) &(0.90, 0.80) &(0.90, 0.70)  &(0.70, 0.70)&(0.90, 0.90) \\ \hline
$\eta_2^*, ~\eta_3^*$ & 36.4870, ~5.0767 & 7.2923, ~0 & 0, ~0 &  NA \\
$\text{Eff}_1({\bf w}^{*m})$ & 0.8694  & 0.9360  & 1.000  &  NA \\
$\text{Eff}_2({\bf w}^{*m}), ~\text{Eff}_3({\bf w}^{*m})$
& 0.9000, ~0.8000  & 0.9000, ~0.7035  &  0.7317, ~0.7746 & NA \\ 
\hline
\end{tabular}
\end{center}
\label{table2}
\end{small}
\end{table}

Computing the optimal designs for $(m_2, m_3) = (0.9, 0.8)$ took 19.5, 25.6, and 36.2 seconds for $N = 101, 501, 1001$. Verifying the optimality of the efficiency constrained design took less than a second.

\vspace{2mm}

\noindent
{\bf Application 2.}  There are several dose response models
commonly used in clinical dose finding studies. We consider
four competing regression models of the form \eqref{Model1} from 
\citet{bretz2010practical} to construct maximin optimal designs  and
verify the necessary and sufficient conditions in Theorem \ref{Th2} for the optimal designs.
The four different models are: (i) linear model: $f_1(x, \bm \theta)=\theta_{11} + \theta_{12}x$, (ii) Emax I model: $f_2(x, \bm \theta)=\theta_{21} + \theta_{22}x/(\theta_{23}+x)$, (iii), Emax II model: $f_3(x, \bm \theta)=\theta_{31} + \theta_{32}x/(\theta_{33}+x)$, (iv) Logistic model: $f_4(x, \bm \theta)=\theta_{41} + 
\theta_{42}/\left(1+\exp[(\theta_{43}-x)/\theta_{44}]    \right)$,
where $x \in [0, 500]$ ($\mu$g) is the dose level. Let $S_N$ contain $N = 501$ equally spaced grid points in $[0, 500]$. As in \citet{bretz2010practical}, we assume that the true parameter values for the Emax I , Emax II, and logistic models are, respectively,
 (60, 294, 25),
 (60, 340, 107.14), and
(49.62, 290.51, 150, 45.51). (The
information matrix for the linear model does not depend on its true parameter values.)

We let  $\Phi_k({\bf w}) = \phi_k(\mathcal{I}_{f_k}({\bf w}))$ with $\phi_k$ corresponding to $D$-optimality (defined in Table \ref{tab:eff}) for all $k = 1, \ldots, 4$, and then solve problem \eqref{MaximinProb5} via \verb+CVX+ to obtain a solution denoted as ${\bf w}^{*mm}$. The single-objective and maximin D-optimal designs on $S_N$ are given in Table \ref{table3}. We found that $t^*=1.1712$,  $\text{Eff}_k({\bf w}^{*mm})=0.8538$ for $k=1,2, 4$ and $\text{Eff}_3({\bf w}^{*mm})=0.8547$. Solving problem \eqref{LProg2} via the MATLAB function \verb+linprog+ yielded $\eta_1^* = 0.1983, \eta_2^* = 0.1291, \eta_3^* = 0$, and $\eta_4^* = 0.0968$. Since we obtain a solution,  we know that ${\bf w}^{*mm}$ is the maximin $D$-optimal design (Theorem \ref{Th2}).

Figure \ref{fig4} displays plots of $d_{\phi_k, f_k}(u_i, {\bf w}^{*mm})$ for $k = 1, \ldots, 4$ and \\ $\sum \limits_{k=1}^4 \eta_k^* d_{\phi_k, f_k}(u_j, {\bf w}^{*mm})$. Figure \ref{fig4}(e) confirms that Condition 3 in Theorem \ref{Th2} is satisfied. Figure \ref{fig4}(a)--(d) shows that  ${\bf w}^{*mm}$ is not the single-objective $D$-optimal design for any of the four models (Theorem  \ref{thm:single-extend} and Lemma \ref{lem:other}). 

\begin{table}[h!]
\begin{small}
\caption{Optimal designs for Application 2.}
\begin{center}
\begin{tabular}{ccccc} \hline
linear model & Emax I  & Emax II   & logistic & maximin \\
points (weights)  &  points (weights)  & points (weights)  & points (weights)  & points (weights)  \\ \hline
0 (0.5000)  & 0 (0.3333) & 0 (0.3333) & 0 (0.2500) & 0 (0.2406) \\
500 (0.5000) & 22 (0.3333) & 75 (0.3333) & 114 (0.2500) & 19 (0.1806) \\
                 & 500 (0.3333)  & 500 (0.3333) & 204 (0.1316)  & 112 (0.1314) \\
                 &  &  & 205 (0.2500) & 204 (0.1070) \\
                 & & & 500 (0.2500) &  205 (0.0178) \\
                 & & & &  500 (0.3225) \\ 
\hline
\end{tabular}
\end{center}
\label{table3}
\end{small}
\end{table}

\begin{figure}[h!]
\centering
\includegraphics[scale=0.9]{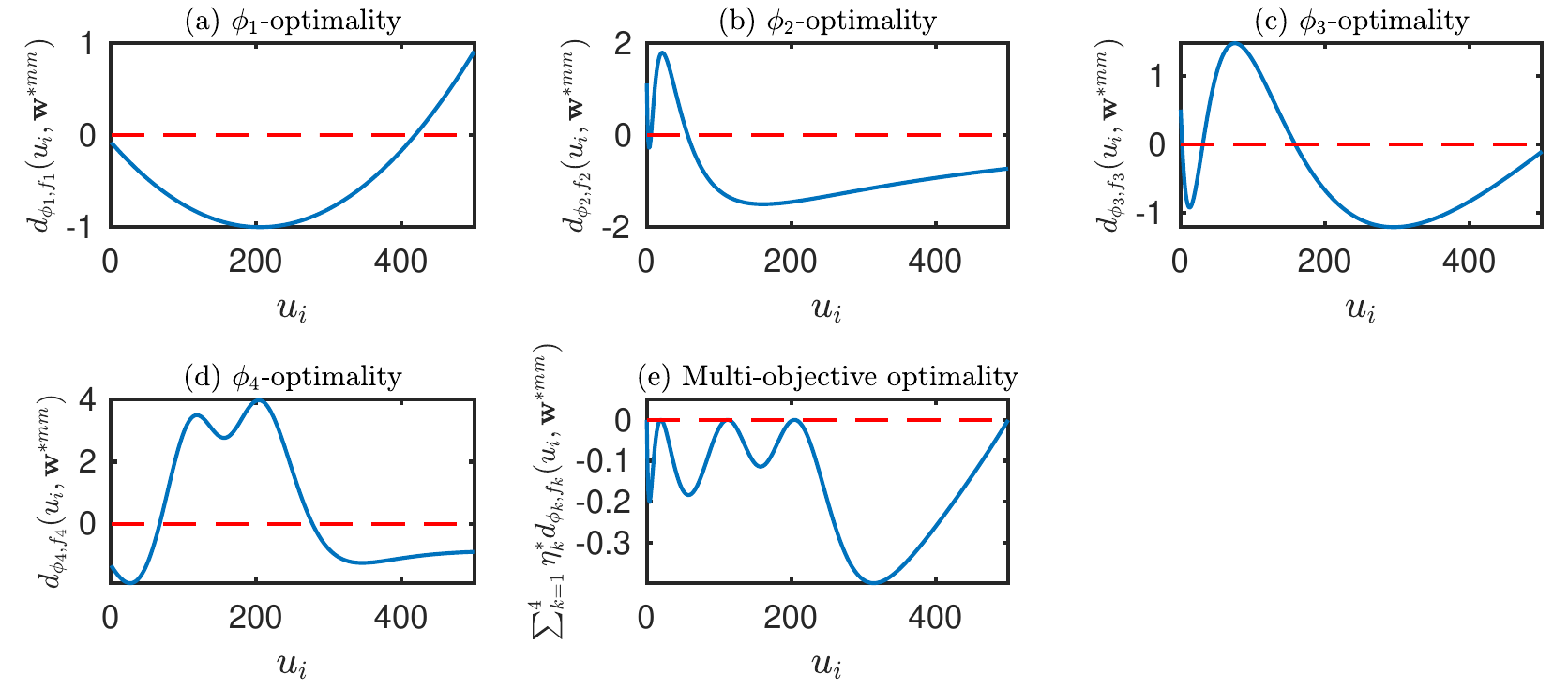}
\caption{For Application 2, we display plots of (a)--(d) $d_{\phi_k, f_k}(u_i, {\bf w}^{*mm})$ for $k = 1, 2, 3, 4$, and (e) $ \sum \limits_{k=1}^4 \eta_k^* d_{\phi_k, f_k}(u_i, {\bf w}^{*mm})$. In panel (e), the red dashed line represents the horizontal line $y = \delta$, for $\delta = 10^{-4}.$ }
\label{fig4}
\end{figure}

Computing the optimal designs took 18.2, 31.0, and 48.7 seconds for $N = 101, 501, 1001$. Verifying the optimality of the maximin design took less than a second.

\noindent {\bf Application 3}: Consider the linear model of the form \eqref{Model1} with $p = 2$, $q = 5$, $f({\bf x}; \bm \theta) = \theta_1 + x_1 \theta_2 + x_2 \theta_3 + x_1x_2 \theta_4 + x_2^2 \theta_5$, and $S = \{0, 1\} \times [-1, 1]$. We let $S^1_{N/2}$ contain 201 equally spaced points on $[-1, 1]$ and discretize the design space $S$ to form $S_N = \{0, 1\} \times S^1_{N/2}$, for $N = 402$. Here, the information matrix does not depend on the true parameter values. 

We let $\Phi_k({\bf w}) =  \phi_k(\mathcal{I}_f({\bf w}))$ for $k = 1, 2, 3$ with $\phi_1$ corresponding to $A$-optimality, $\phi_2$ corresponding to $E$-optimality, and $\phi_3$ corresponding to $c$-optimality with ${\bf c} = (0, 0, 0, 1, 0)^T$, and then solve problem \eqref{MaximinProb5} as described in Section 4.1 to obtain the maximin optimal design; we denote the solution as ${\bf w}^{*mm}_{lm}$. The single-objective and maximin optimal designs on $S_N$ are given in Table \ref{tab:linapp-design}. We found that $t^* = 1.2979$, $\text{Eff}_1({\bf w}^{*mm}_{lm}) = 0.9298$, and $\text{Eff}_1({\bf w}^{*mm}_{lm}) = 0.7705$ for $k = 2, 3$. In this case, the geometric multiplicity of $\lambda_{min}(\mathcal{I}_f({\bf w}^{*mm}_{lm}))$ was equal to one. Thus, to verify the optimality of ${\bf w}^{*mm}_{lm}$, we use the \verb+linprog+ MATLAB function to minimize $\sum \limits_{k=1}^3 \eta_k$ subject to the linear inequality constraints in \eqref{eq:defde-mm} and  the linear equality and inequality constraints in \eqref{eq:Cond1m} -- \eqref{eq:Cond2mb}.  We found a solution at $\eta_1^* = 0$, $\eta_2^* = 0.2445$, and $\eta_3^* = 0.0151$. Therefore, the maximin design we found is optimal. 

\begin{table}[h!]
\begin{small}
\caption{Optimal designs for Application 3. }
\begin{center}
\begin{tabular}{ccccc} \hline
$A$-optimal & $E$-optimal  & $c$-optimal   & maximin  \\
points [weights]  &  points [weights] & points [weights]  & points [weights] \\ \hline
(0, -1) [0.1859] & (0, -1) [0.2069] & (0, -1) [0.2500] & (0, -1) [0.1926] \\
(1, -1) [0.1399] & (1, -1) [0.1379] & (1, -1) [0.2500] & (1, -1) [0.1926] \\
(0, 0) [0.2287] & (0, 0) [0.2414] & (0, 1) [0.2500] & (0, 0) [0.1679]\\ 
(1, 0) [0.1197] & (0, 1) [0.0690] & (1, 1) [0.2500] &  (1, 0) [0.0616] \\ 
(0, 1) [0.1859] & (1, 1) [0.2069] & & (0, 1) [0.1926]\\ 
(1, 1) [0.1399] & & & (1, 1) [0.1926]  \\ 
\hline
\end{tabular}
\end{center}
\label{tab:linapp-design}
\end{small}
\end{table}
Computing the optimal designs took 11.2, 14.5, and 23.2 seconds for $N/2 = 101, 201, 401$. Verifying the optimality of the maximin design took less than a second.

\section{Conclusion}

In this paper, we show how to solve multi-objective optimal design problems on a discrete design space using convex optimization, and how to verify the optimality of the designs using linear programming. Our approach can be applied to efficiency-constrained or maximin optimal design problems that combine any of the single-objective criteria in Table \ref{tab:eff}.\textcolor{white}{\citep{lewis1999nonsmooth} }

The multi-objective optimal design setting offers a natural opportunity to gain robustness against parameter and/or model misspecification: we can include objective functions formulated with a range of guesses for $\bm \theta^*$ and/or objective functions formulated with the information matrices of multiple models. A sequential multi-objective optimal design setting may offer further opportunities to gain robustness, as we could select design points and weights in stages, and use the data from each stage to inform the choice of parameters and/or models used in the objective functions for the next stage. This may provide a fruitful avenue for future work.

We were able to achieve the results and algorithms in this paper because the inverse of the  asymptotic covariance matrix of the ordinary least squares estimator under model \eqref{Model1} is a linear function of ${\bf w}^*$; see equation \eqref{Info1}. It would be straightforward to extend the results and algorithms in this paper to other models and estimators that have a similar property. For example, we could allow the vector of errors in \eqref{Model1} to be heteroskedastic or have a block diagonal covariance structure, and use the generalized least squares estimator. Another example is generalized linear models with a canonical link function, where we estimate $\bm \theta$ with the maximum likelihood estimator.

Necessary and sufficient conditions for optimality that involve a set of unknown parameters appear in contexts outside of the multi-objective design problems we consider in this paper. For example, this is the case for any single-objective optimal design problem involving a convex but non-differentiable objective function (e.g. single-objective $E$-optimality). When the conditions define linear equalities and inequalities in these unknown parameters, we can verify them using linear programming as in this paper. 

A limitation of our results and algorithms is the assumption of a discrete design space.  An important direction for future work is to develop results and algorithms under a continuous design space. \textcolor{white}{\citep{clarke1983optimization}
}
\section*{Acknowledgements}

This research work was partially supported by Discovery Grants from the Natural Sciences and Engineering Research Council of Canada.

\par


\bibhang=1.7pc
\bibsep=2pt
\fontsize{9}{14pt plus.8pt minus .6pt}\selectfont
\renewcommand\bibname{\large \bf References}
\expandafter\ifx\csname
natexlab\endcsname\relax\def\natexlab#1{#1}\fi
\expandafter\ifx\csname url\endcsname\relax
  \def\url#1{\texttt{#1}}\fi
\expandafter\ifx\csname urlprefix\endcsname\relax\def\urlprefix{URL}\fi

  \bibliographystyle{chicago}      
  \bibliography{refs}   

\appendix

    \setcounter{prop}{0}
    \renewcommand{\theprop}{\Alph{section}\arabic{prop}}
    
\section*{Appendix}

\section{General convex optimization theory} 
In this section, we review general results on optimality conditions in convex optimization.

\begin{prop} 
\label{prop:general}
Let $\Phi: \mathbb{R}^N \mapsto \mathbb{R}$ be a convex function and $C$ be a closed convex set. Then, ${\bf w}^* \in \underset{{\bf w} \in C}{\arg \min} ~ \Phi({\bf w})$ if and only if there exists ${\bf g} \in \partial \Phi({\bf w}^*)$ such that 
\begin{equation} 
{\bf g}^T ({\bf w} - {\bf w}^*) \geq 0, ~~~ \text{for all } {\bf w} \in C, 
\end{equation} 
where $\partial \Phi({\bf w}^*) \equiv \{ {\bf g} \in \mathbb{R}^N: \Phi({\bf w)} - \Phi({\bf w}^*) \geq {\bf g}^T ({\bf w} - {\bf w}^*) ~ \forall ~ {\bf w} \in {\mathbb{R}^N}\}$.
\end{prop} 
Proposition \ref{prop:general} is a direct consequence of Theorem 4.14 of Mordukhovich and Nam (2013).  The following result characterizes optimality for constrained convex optimization problems. 
\begin{prop} 
\label{prop:general-constrained}
Define the following convex optimization problem: 
\begin{equation}\underset{{\bf w} \in C}{\min} ~~~\Phi_0({\bf w}) 
\qquad \mbox{subject to:} ~\Phi_l({\bf w}) \leq 0, ~l=1, \ldots, L,\qquad \qquad 
\label{eq:constrained}
\end{equation}
where $\Phi_0, \ldots, \Phi_L:\mathbb{R}^N\rightarrow \mathbb{R}$ are  convex functions and $C$ is a closed convex set. Suppose that Slater's condition holds, i.e. there exists ${\bf w}' \in C$ such that $\Phi_l({\bf w}') < 0$ for all $l=1, \ldots, L$. Then, a feasible solution ${\bf w}^*$ of \eqref{eq:constrained} solves \eqref{eq:constrained} if and only if there exists $\eta_1, \ldots, \eta_L \geq 0$ such that $\eta_l \Phi_l({\bf w}^*) = 0$ for all $l = 1, 2, \ldots, L$ and ${\bf w}^* \in \underset{w \in C}{\arg \min} \left \{ \Phi_0({\bf w}) + \sum \limits_{l=1}^L \eta_k \Phi_l({\bf w}) \right \}$. 
\end{prop} 
\begin{proof} 
We assumed that Slater's condition holds. Thus, Theorem 4.18 of Mordukhovich and Nam (2013) says that ${\bf w}^*$ is optimal for \eqref{eq:constrained} if and only if there exist multipliers $\eta_1, \ldots, \eta_L \geq 0$ such that $\eta_l \Phi_l({\bf w}^*) = 0$ for all $l = 1, 2, \ldots, L$ and
\begin{equation} 0 \in \partial \Phi_0({\bf w}^*) + \sum \limits_{l=1}^L \eta_l \partial \Phi_l({\bf w}^*) + \left \{ {\bf g} \in \mathbb{R}^N:  {\bf g}^T {\bf w}^* \geq {\bf g}^T {\bf w} ~ \forall ~ {\bf w} \in C \right \}.\label{eq:subdiff}
\end{equation} 
Thus, by Lemma \ref{lem:sub}, \eqref{eq:subdiff} is equivalent to: 
\begin{equation} 
\exists {\bf g} \in \partial \left (\Phi_0 + \sum \limits_{l=1}^L \eta_l \Phi_l \right )({\bf w}^*) \text{~ such that ~} {\bf g}^T ({\bf w} - {\bf w}^*) \geq 0 \quad { \mbox{ for all } {\bf w} \in C}. \label{eq:subdiff2}
\end{equation} 
Finally, it follows from Proposition \ref{prop:general} that \eqref{eq:subdiff2} holds if and only if ${\bf w}^* \in \underset{{\bf w} \in C}{\arg \min} \left \{ \Phi_0({\bf w}) + \sum \limits_{l=1}^L \eta_l \Phi_l({\bf w}) \right \}$. 
\end{proof} 

\section{Proof of Theorem \ref{thm:single-extend}}
\label{sec:proof-single}

Let ${\bf w}^* \in \Omega$. 
Since $\Phi({\bf w}) = \phi({\cal I}_f({\bf w}))$ is convex
and $\Omega$ is a closed convex set, from Proposition \ref{prop:general},  ${\bf w}^* \in \underset{{\bf w} \in \Omega}{\arg \min} ~ \Phi({\bf w})$ if and only if
\begin{equation} \exists {\bf g} \in \partial \Phi({\bf w}^*) \text{ such that } {\bf g}^T ({\bf w} - {\bf w}^*) \geq 0, \text{ for all } {\bf w} \in \Omega. \label{eq:nobasis} 
\end{equation}  
Suppose that \eqref{eq:nobasis} holds. Since ${\bf e}_i \in \Omega$ for all $i = 1, 2, \ldots, N$, we have that \eqref{eq:ineq-extend} holds. Now suppose that \eqref{eq:ineq-extend} holds. Then, because all ${\bf w} \in \Omega$ have $w_i \geq 0$ for all $i = 1, 2, \ldots, N$, \eqref{eq:ineq-extend} implies that there exists ${\bf g} \in \partial \Phi({\bf w}^*)$ such that
$ \sum \limits_{i=1}^N w_i {\bf g}^T ({\bf e}_i - {\bf w}^*) \geq 0, \text{ for all } {\bf w} \in \Omega.$
Observing that ${\bf g}^T ({\bf w} - {\bf w}^*) =  \sum \limits_{i=1}^N w_i {\bf g}^T ({\bf e}_i - {\bf w}^*)$ completes the proof. \qedsymbol{}

\section{Proof of Lemma \ref{prop:eopt}}
 
We will start by establishing properties of $(-\Phi)({\bf w}) = \lambda_{min}(\mathcal{I}_f({\bf w}))$, as this allows us to take advantage of existing theoretical results about the minimum eigenvalue function $\lambda_{min}({\bf M})$. 

The function $(-\Phi)({\bf w})$ is not differentiable at every point in $\mathbb{R}^N$. Furthermore, the notion of a subdifferential does not apply to $(-\Phi)({\bf w})$, as $(-\Phi)({\bf w})$ is a \emph{concave} rather than a convex function. However, it has a \emph{Clarke subdifferential} \citep{clarke1983optimization}, which generalizes the notion of the gradient to the class of locally Lipschitz continuous functions. The Clarke subdifferential of a locally Lipschitz continuous function $h({\bf w})$ on $\mathbb{R}^N$ is defined as: 
$$ \partial^C h({\bf w}) = \text{co} \left ( \left \{ {\bf v} \in \mathbb{R}^N: \exists ~\{{\bf w}_k\}_{k=1}^\infty \text{ s.t.} \underset{k \rightarrow \infty}{\lim} {\bf w}_k \text{ exists},  \nabla h({\bf w}_k) \text{ exists, and } \underset{k \rightarrow \infty}{\lim} \nabla h({\bf w}_k) = {\bf v} \right \} \right ), $$ 
where $\text{co}(S)$ is the convex hull of the set $S$, i.e. the intersection of all convex sets containing $S$. 

We can characterize the Clarke subdifferential $\partial^C (-\Phi)({\bf w})$ by observing that $(-\Phi)({\bf w})$ is the composition of the non-differentiable concave function $\lambda_{min}$ with the linear function $\mathcal{I}_f$. Thus, it follows from the Clarke subdifferential chain rule (Theorem 2.3.10 in \citealt{clarke1983optimization}) that 
${\bf g}_c \in \partial^c (-\Phi)({\bf w}^*)$ if and only if there exists ${\bf M} \in \partial^c \lambda_{\min}(\mathcal{I}_f({\bf w}^*))$ such that for any ${\bf w} \in \mathbb{R}^N$, 
\begin{equation} {\bf g}_c^T {\bf w} = \text{trace}\left ({\bf M} \left ( \sum \limits_{i=1}^N w_i {\bf z}_f({\bf u}_i){\bf z}^T_f({\bf u}_i) \right ) \right ). \label{eq:lambdamin1} 
\end{equation} 
Furthermore, Corollary 10 of \citet{lewis1999nonsmooth} says that 
\begin{align} \partial^c \lambda_{min} (\mathcal{I}_f({\bf w}^*)) 
&= \left \{ \sum \limits_{j=1}^{r^*} a_j {\bf v}_j^* [{\bf v}_j^*]^T: a_j \geq 0, \sum \limits_{j=1}^{r^*} a_j = 1 \right \}, \label{eq:lambdamin2}
\end{align} 
recalling that in the statement of Lemma \ref{prop:eopt} we defined  ${\bf v}_1, \ldots, {\bf v}_{r^*}$ to be an arbitrary set of ${r^*}$ linearly independent unit eigenvectors associated with $\lambda_{min}(\mathcal{I}_f({\bf w}^*))$ where ${r^*}$ is the geometric multiplicity of $\lambda_{min}(\mathcal{I}_f({\bf w}^*))$. It follows from the definition of $\mathcal{I}_f({\bf w}^*)$ in \eqref{Info1} and \eqref{eq:lambdamin1}--\eqref{eq:lambdamin2} that if ${\bf g}_c \in \partial^c (-\Phi)({\bf w}^*)$, then there exists $a_1, \ldots, a_{r^*} \geq 0$ such that $\sum \limits_{j=1}^{r^*} a_j = 1$ and 
${\bf g}_c^T ({\bf e}_i - {\bf w}^*) = \sum \limits_{j=1}^{r^*} a_j([{\bf v}_j^*]^T {\bf z}_f({\bf u}_i))^2 - \lambda_{\min}(\mathcal{I}_f({\bf w}^*)), \text{ for all } i = 1, 2, \ldots, N. $
Furthermore, $\partial \Phi({\bf w}^*) = \partial^c \Phi({\bf w}^*) = - \left [ \partial^c (-\Phi) ({\bf w}^*) \right ],  $
where the first equality follows from Proposition 2.2.7 of \citep{clarke1983optimization}, and the second equality follows from Proposition 2.3.1 of \citep{clarke1983optimization}. Therefore, for any ${\bf g} \in \partial \Phi({\bf w}^*)$, we know that $-{\bf g} \in \partial^c (-\Phi)({\bf w}^*)$. Thus, there exists $a_1, \ldots, a_{r^*} \geq 0$ such that $\sum \limits_{j=1}^{r^*} a_j = 1$ and 
$ -{\bf g}^T ({\bf e}_i - {\bf w}^*) = \sum \limits_{j=1}^{r^*} a_j([{\bf v}_j^{*}]^T {\bf z}_f({\bf u}_i))^2 - \lambda_{\min}(\mathcal{I}_f({\bf w}^*)), \text{ for all } i = 1, 2, \ldots, N.$ 
~~~\qedsymbol{}

\section{Proof of Theorem \ref{Th1}}
Since we assumed that there exists ${\bf w} \in \Omega$ satisfying $\text{Eff}_k({\bf w}) > m_k$ for all $k=2, \ldots, K$, we have that Slater's condition holds for problem \eqref{EffProb2}. Furthermore, for all $k = 1, 2, \ldots, K$, $\Phi_k({\bf w}) =  \phi_k(\mathcal{I}_{f_k}({\bf w}))$ is a convex function
and $\Omega$
is a convex set. Thus, it follows from Proposition \ref{prop:general-constrained} that ${\bf w}^* \in \Omega$ solves \eqref{EffProb2} if and only if there exists $\eta_2, \ldots, \eta_K \geq 0$ such that $\eta_k (\Phi_k({\bf w}^*) - h_k(m_k)) = 0$ for all $k = 2, 3, \ldots, K$ and 
${\bf w}^* \in \underset{{\bf w} \in \Omega}{\arg \min}  \left \{ \Phi_1({\bf w}) + \sum \limits_{k=2}^K \eta_k (\Phi_k({\bf w}) - h_k(m_k)) \right \}.$ Observing that 
$\sum \limits_{k=2}^K \eta_k h_k(m_k)$ does not depend on ${\bf w}$ completes the proof. 
\qedsymbol{}

\section{Proof of Theorem \ref{Th2}}
We first confirm that the following restatement of \eqref{MaximinProb5}, 
\begin{equation}
\underset{{\bf w} \in \Omega,t \geq 0}{\min} ~~t \qquad 
\mbox{subject to:}~ \Phi_k({\bf w})  \le h_k(1/t)  , ~k=1, \ldots, K,
\label{MaximinProb6}
\end{equation}
satisfies the conditions in Proposition \ref{prop:general-constrained}. Since $\Omega$ and $\mathbb{R}$ are closed convex sets, $\Omega \times \mathbb{R}$ is a closed convex set. Furthermore, for all $k = 1, 2, \ldots, K$, $\Phi_k({\bf w}) =  \phi_k(\mathcal{I}_{f_k}({\bf w}))$ is a convex function.
We also know that for all $k = 1, 2, \dots, K$, $h_k(1/t)$ in \eqref{eq:defh} is a concave function of $t$ for all $k = 1, 2, \ldots, K$, as $-\lambda_{\min}({\bf M})/t$ is a concave function of $t$ for any positive definite matrix ${\bf M}$, $t$ is a linear function, and $q_k \log(t)$ is a concave function. To show that Slater's condition holds, we need to find ${\bf w}' \in \Omega$ and $t' > 0$ with $\Phi_k({\bf w}') < h_k(1/t')$ for all $k = 1, 2, \ldots, K$. It follows from the definition of the efficiency functions $\text{Eff}_k({\bf w}')$ in Table \ref{tab:eff} that $\Phi_k({\bf w}') < h_k(1/t')$ if and only if $\text{Eff}_k({\bf w}') > 1/t'$, and that $\text{Eff}_k({\bf w}') > 0$ for all ${\bf w}' \in \Omega$. Thus, choosing ${\bf w}' = \underset{{\bf w} \in \Omega}{\arg \min} ~\Phi_1({\bf w})$ and $t' = 2/ \left ( \underset{k = 1, 2, \ldots,K}{\min} ~ \text{Eff}_k ({\bf w'}) \right)$ satisfies Slater's condition. 
 
We can now apply Proposition \ref{prop:general-constrained} to  \eqref{MaximinProb6} to yield the following result: a feasible solution $({\bf w}^*, t^*)$  for \eqref{MaximinProb6} solves problem \eqref{MaximinProb6} if and only if there exists $\nu, \eta_1, \ldots, \eta_K \geq 0$ satisfying: 
\begin{gather} \nu t^* = 0, \label{eq:proofcond1} \\ 
\eta_k (\Phi_k({\bf w}^*) - h_k(1/t^*)) = 0 \text{ for all }k = 1, 2, \ldots, K, \label{eq:proofcond2} \\ 
({\bf w}^*, t^*) \in \underset{{\bf w} \in \Omega, t \in \mathbb{R}}{\arg \min} \left \{t - \nu t + \sum \limits_{k=1}^K \eta_k (\Phi_k({\bf w}) - h_k(1/t)) \right \}. \label{eq:proofcond3} 
\end{gather} 
The optimization problem in \eqref{eq:proofcond3} is separable. Thus, \eqref{eq:proofcond3} can be rewritten as 
\begin{gather} 
{\bf w}^* \in \underset{{\bf w} \in \Omega}{\arg \min} \left \{ \sum \limits_{k=1}^K \eta_k \Phi_k({\bf w}) \right \}, \label{eq:proofcond4} \\ 
t^* \in \underset{t \in \mathbb{R}}{\arg \min} \left \{ t - \nu t - \sum \limits_{k=1}^K \eta_k h_k(1/t) \right \}.\label{eq:proofcond5}
\end{gather} 
Since $g(t) = t - \nu t - \sum \limits_{k=1}^K \eta_k h_k(1/t)$ is a convex function, we can rewrite \eqref{eq:proofcond5} as 
\begin{equation} 
1 - \nu - \sum \limits_{k=1}^K \eta_k \left [ \frac{d}{dt} h_k(1/t) \Big |_{t = t^*} \right ] = 0. \label{eq:proofcond6} 
\end{equation} 
This means that there exists  $\nu, \eta_1, \ldots, \eta_K \geq 0$ satisfying \eqref{eq:proofcond1}--\eqref{eq:proofcond3} if and only if there exists $\eta_1, \ldots, \eta_K \geq 0$ satisfying \eqref{eq:proofcond2}, \eqref{eq:proofcond4}, and 
\begin{gather} 
1  - \sum \limits_{k=1}^K \eta_k \left [ \frac{d}{dt} h_k(1/t) \Big |_{t = t^*} \right ]  \geq 0, \quad \quad t^* \left ( 1  - \sum \limits_{k=1}^K \eta_k \left [ \frac{d}{dt} h_k(1/t) \Big |_{t = t^*} \right ] \right ) = 0. \label{eq:proofcond8}
\end{gather} 
Since \eqref{eq:proofcond2} is Condition 2 in Theorem \ref{Th2}, and \eqref{eq:proofcond4} is Condition 3 in Theorem \ref{Th2}, it remains to show that \eqref{eq:proofcond8} is equivalent to Condition 1 in Theorem \ref{Th2}.
It suffices to show that $t^* > 0$. Recall that \eqref{MaximinProb6} is equivalent to \eqref{MaximinProb4}, so the optimal value for $t$ in \eqref{MaximinProb6} is the reciprocal of the maximin efficiency attained by the optimal design. Efficiencies are bounded between 0 and 1, so the optimal value for $t$  must be greater than 1, i.e. $t^* > 1$.  \qedsymbol{}

\end{document}